\documentclass[11pt]{article}
\usepackage[totalwidth=15.0cm,totalheight=20.0cm]{geometry}
\usepackage{latexsym,amsthm,amsmath,amssymb}
\usepackage[ruled, linesnumbered]{algorithm2e}

\usepackage{times}  
\usepackage{helvet}  
\usepackage{courier}  
\usepackage[hyphens]{url}  
\usepackage{graphicx} 
\usepackage{latexsym,amsthm,amsmath,amssymb,url}
\usepackage{todonotes}

\usepackage{enumitem}
\usepackage{tikz,authblk}
\usetikzlibrary{decorations.pathreplacing}
\usepackage{caption}
\usepackage{subcaption}
\usepackage{hyperref}
\usepackage{color}

\usepackage{tikz,authblk}
\usetikzlibrary{decorations.pathreplacing}

\newcommand{\cA}{\ensuremath{\mathcal{A}}}
\newcommand{\cB}{\ensuremath{\mathcal{B}}}

\newcommand{\cI}{\ensuremath{\mathcal{I}}}

\newcommand{\bigO}{\mathcal{O}}

\newtheorem{theorem}{Theorem}[section]

\newtheorem{proposition}[theorem]{Proposition}

\newtheorem{claim}{Claim}[theorem]
\newtheorem{observation}{Observation}[theorem]

\newcommand{\EE}[1]{{#1}}
\newcommand{\GG}[1]{{#1}}

\newcommand{\MW}[1]{{#1}}
\newcommand{\AY}[1]{{#1}}

\begin{document}

\title{Preference Swaps for the Stable Matching Problem}
\author[1]{Eduard Eiben}
\author[1]{Gregory Gutin}
\author[2]{Philip R. Neary}
\author[3]{Cl\'{e}ment Rambaud}
\author[1]{Magnus Wahlstr{\"o}m}
\author[4,5]{Anders Yeo}
\affil[1]{Department of Computer Science, Royal Holloway, University of London, TW20 0EX, Egham, Surrey, UK}
\affil[2]{Department of Economics, Royal Holloway, University of London, TW20 0EX, Egham, Surrey, UK}
\affil[3]{DIENS, \'Ecole normale sup\'erieure, PSL University, Paris, France}
\affil[4]{IMADA, University of Southern Denmark, Odense, Denmark} 
\affil[5]{Department of Mathematics, University of Johannesburg, Auckland Park, 2006 South Africa}

\date{}
\maketitle

\begin{abstract}
\noindent 
An instance $I$ of  the Stable Matching Problem (SMP) is given by a bipartite graph with a preference list of neighbors for every vertex. 
A swap in $I$ is the exchange of two consecutive vertices in a preference list.  A swap can be viewed as a smallest perturbation of $I$. 
Boehmer et al. (2021) designed a polynomial-time algorithm to find the minimum number of swaps required to turn a given maximal matching into a stable matching. We generalize this result to the many-to-many version of SMP. We do so first by introducing a new representation of SMP as an extended bipartite graph and subsequently by
reducing the problem to submodular minimization. It is a natural problem to establish the computational complexity of
deciding whether at most $k$ swaps are enough to turn $I$ into an instance where one of the maximum matchings is stable. 
Using a hardness result of Gupta et al. (2020), we prove that this problem is NP-hard and, moreover, this problem parameterised by $k$ is W[1]-hard.
We also obtain a lower bound on the running time for solving the problem using the Exponential Time Hypothesis.\end{abstract}

\pagestyle{plain}

\section{Introduction}

We consider the classic stable matching problem of \cite{gale1962college} described as follows. There is a two-sided market with two types of agents: workers on one side and firms on the other. Workers seek employment while firms have vacancies and so are looking to hire \GG{(let us assume for now that every firm wishes to employ just one worker and vice versa)}.  
Workers have preference lists over firms and firms have preference lists over workers. This is modeled as a bipartite graph $G$, where the partite sets are workers $W$ and firms $F$ and there is a preference list for each $a\in W\cup F$ consisting of all neighbors of $a$ in $G$. A matching $M$ of workers to firms in $G$ is said to be stable if there is no edge $wf$ in $G$ ($w\in W$, $f\in F$) such that $w$ is either not matched in $M$ or matched to a firm less preferable to $w$ than $f$, and $f$ is either not matched in $M$ or matched to a worker less preferable to $f$ than $w$.
In the {\sc Stable Matching Problem} (SMP), given a bipartite graph $G$ and a set $L_G$ of preference lists, the goal is to find a stable matching.


The deferred acceptance algorithm of Gale and Shapley \cite{gale1962college} confirms that a stable matching always exists. By definition, every stable matching is maximal, but in general not every maximal matching is stable.
In this paper we focus on the relationship between maximal matchings and stable matchings. In particular, when a maximal matching is not a stable matching, how far from being stable is it?

%

By virtue of the Rural Hospitals Theorem \cite{McVitieW70,Roth84,Roth86}, the set of agents matched is the same in every stable matching. Therefore, to get a larger stable matching, the preference lists need to be amended/altered.
Consider a worker who ranks firm $x$ just ahead of firm $y$ in their preference list. A \emph{swap} exchanges the relative placing of firms $x$ and $y,$ leaving the remainder of the preference list unchanged. (A similar statement holds for the preference lists of firms.) With swaps so defined, we can now address questions like that above concerning the distance between stable matchings and maximal matchings, where the unit of distance we employ is that of a swap.



 
The economic motivation for this paper concerns how a policy maker might amend an existing two-sided matching market (a so-called market without prices) so that the resulting stable outcomes are more socially desirable. For example, when a maximum matching is a stable matching the number of unemployed individuals and the number of unfilled vacancies is minimised.
Historically, markets with prices were allowed to evolve naturally, with policy makers only interfering when they saw ways to `improve' outcomes.\footnote{Examples of tools available to policy makers in such markets include taxes, tariffs, quotas, outright bans, price ceilings, etc.} But studies that address how market interference might improve existing matching markets are surprisingly scarce.\footnote{This should not be confused with the enormous literature on \emph{market design}, that seeks to engineer from scratch matching markets with `desirable' properties.}

Above for simplicity we  assumed that every firm wishes to employ just one worker and every worker wants only one job. This is an example of so-called one-to-one SMP.  Such a model is restrictive as, for example, it does not allow us to consider natural situations like that where firms have more than one vacancy and/or a worker may seek more than one job \cite{Roth84,Roth86,KamadaKojima:2015}.

\noindent
\textbf{Contributions.} We study two related problems on using swaps to obtain large (maximum) stable matchings: 
\begin{description}
\item[Problem 1:] Decide whether at most $k$ swaps can turn a given maximal matching into a stable matching.
\item[Problem 2:] Decide whether at most $k$ swaps are enough to have an arbitrary stable {\em maximum} matching.
\end{description}
While the two problems look similar at the first glance, their computational complexities differ. For Problem 1, we consider many-to-many SMP and obtain a polynomial-time algorithm, see  Theorem \ref{thm:capacitated_algorithm}. 
Our design of the algorithm uses two tools:
(a) a new representation of SMP input as an extended bipartite graph ${\cal G}_G$ which encodes both $G$ and its preference lists and provides us with a convenient way to study swaps, and
(b) a reduction to submodular function minimization. 

Theorem \ref{thm:capacitated_algorithm} may suggest that Problem 2 is also  polynomial-time solvable. Unfortunately, this is highly unlikely as Theorem \ref{thm:hardness_swap_distance} shows that the problem is  NP-hard and moreover if $k$ is the parameter then the problem is W[1]-hard\footnote{For an excellent introduction to parameterized complexity, see \cite{CyganFKLMPPS15}.}, 
even if $G$ has a perfect matching. We also obtain a lower bound on the running time for solving the problem using the Exponential Time Hypothesis (ETH) \cite{ImpagliazzoP99}. 



\noindent
\textbf{Related work.} \GG{The classical formulation of SMP uses agents and their preference lists, or equivalently bipartite graphs  and preference lists of the vertices. 
Maffray  \cite{Maffray1992} gave another representation of SMP using directed graphs, which does not use explicit preference lists. Maffray's representation is better suitable for some SMP studies than the standard one  \cite{BalinskiR97,BalinskiR98,GutinNY21}. Our new representation of SMP uses extended bipartite graphs and no preference lists either.}

\MW{Boehmer et al. \cite{BoehmerBHN20} give examples of external
manipulations which can lead to changes in one-to-one SMP preference
lists, among them the notion of a swap.  They study several problems
related to matchings that are stable after the application of a small
number of such manipulations steps, including proving that Problem 1
is polynomial-time solvable for one-to-one SMP.  Thus,
Theorem~\ref{thm:capacitated_algorithm} generalizes their result.}

Chen et al. \cite{ChenSS19} study the stable matching $k$-robustness problem for one-to-one SMP, which is the problem of deciding whether $(G,L_G)$ has a stable matching remaining stable even after at most $k$ swaps. Surprisingly, the problem turns out polynomial-time solvable. Chen et al. \cite{ChenSS19} observed that Problem 2 is in XP if $k$ is the parameter and proved that Problem 2 is W[1]-hard if $G$ has a perfect matching and the parameter is $n_u$ the number of vertices unmatched in any stable matching of $G$. 
\GG{By Propositions~\ref{prop:comparison_to_chen_at_all}  and \ref{prop:comparison_to_chen_at_all2}, our W[1]-hardness result is strictly stronger than that in \cite{ChenSS19}.}


Mai and Vazirani \cite{MaiV18,MaiV2020} studied one-to-one SMP when the number of vertices in both partite sets of $G$ is the same, denoted by $n.$ A {\em shift} is an exchange of any two (not necessarily consecutive) vertices in a preference list. 
Mai and Vazirani \cite{MaiV18} considered a given distribution among all shifts and studied the effect of a random shift on matching stability. They obtained a polynomial-time algorithm for finding a matching which is stable in the input instance of SMP and has the maximum probability of remaining stable after a random shift took place. This algorithm applies to the problem we study in this paper only for $k=1$ as a sequence of two swaps may not be a shift. Mai and Vazirani \cite{MaiV2020} considered a set $S$ of permutations in a single preference list and designed an $O(|S|{\rm poly}(n))$-time algorithm for finding a matching which is stable in the given instance of SMP and remains so after any permutation in $S$ took place. 

Bredereck et al. \cite{BredereckCKLN20} designed a polynomial-time algorithm for the following problem: given a bipartite graph $G$, a pair $L$ and $L'$ of preference lists for $G$, a stable matching $M$ of $(G,L)$ and a natural number $k$, decide whether $(G,L')$ has a stable matching $M'$ such that the symmetric difference of $M$ and $M'$ is at most $k.$ 



A way to increase the size of stable matchings is to relax the zero blocking pair condition.
Gupta et al. \cite{gupta2020parameterized} studied the almost stable matching problem, which is a problem of finding a matching whose size exceeds that of a stable matching by at least $t$ and has at most $k$ blocking edges. Let $d$ be the maximum length of a preference list. 
Gupta et al. \cite{gupta2020parameterized} proved the following surprising result: the problem is intractable even when parameterized by the combined parameter $k+t+d.$ This result demonstrates significant computational difficulty in solving natural approximating stable matching problems.

\section{Terminology, Notation and Preliminaries}\label{sec:another_representation}

In this section, we will first give a more formal definition of the {\sc Stable Matching Problem} (SMP) using bipartite graphs and preference lists for vertices and then introduce an equivalent new representation of SMP which also uses bipartite graphs, but without explicit preference lists. The new representation gives us a more suitable approach than the standard one for some SMP problems studied in this paper.

The input of SMP  is $(G,L_G)$, where 
$G=(A,B;E)$ is a bipartite graph with partite sets $A$ and $B$ and edge set $E$ and $L_G$ is the set of lists $\ell(v)=(u_1,\dots ,u_{d(v)})$ which order all the neighbors of every $v\in A\cup B.$ 
If $i<j$ we say that $v$ {\em prefers} $u_i$ to $u_j.$  
Let $M$ be a matching in $G$. Then an edge $ab\in E\setminus M$ with $a\in A,b\in B$ is a {\em blocking edge} of $M$ if the following two conditions hold simultaneously: (i) either $a$ is not an endpoint of any edge in $M$ or there is a $b'\in B$
such that $ab'\in M$ and {$a$ prefers $b$ to $b'$
and, similarly, (ii) either $b$ is not an endpoint of any edge in $M$ or there is an $a'\in A$ such that $a'b\in M$ and {$b$ prefers $a$ to $a'$}. 
Figure \ref{fig:reduction_vertex_cover1} shows preference lists for the vertices of the depicted graph; $a_1$ prefers $b_1$ to $b_3$ and $b_3$ to $b_2$; $(1,3,2)$ only includes indexes of the corresponding vertices in $B$. Figure \ref{fig:reduction_vertex_cover1} also gives examples of blocking edges. 

A matching $M$ is {\em stable} if it has no blocking edges.
It follows from the definition of a blocking edge that a stable matching $M$ is {\em maximal} i.e. no edge from $G$ can be added to $M$ such that it remains a matching. A matching $M$ is {\em perfect} if $|M|=|A|=|B|.$   The aim of SMP is to find a stable matching in $G$. The deferred acceptance algorithm of 
\cite{gale1962college} outputs a stable matching in $G$. In particular,  \cite{gale1962college} studied the above wherein (i) the bipartite graph $G$ is complete, and (ii) the partite sets are of equal size ($|A| =|B|$). 
For an instance of SMP, a {\em swap} $\sigma$ is an adjacent transposition in $\ell(v)$ for some $v\in A\cup B$ i.e. this operation results in exchanging two consecutive elements $u_{i-1},u_{i}$ of $\ell(v)$ and is denoted by
$u_{i-1} \leftrightarrow u_i$. We denote by $\sigma(L_G)$ the new set of preference lists after the swap $\sigma$.


Let us now introduce the new SMP representation, which fuses $G$ and $L_G$ into a single bipartite graph, called the {\em extended bipartite graph}, with specifically structured partite sets as follows.
Let $G=(A,B;E),$ where $A=\{a_1,\dots ,a_{n}\}$ and $B=\{b_1,\dots ,b_{n'}\}.$ Let $d_i$ ($d'_i$, respectively) be the degree of $a_i$ ($b_i$, respectively). 
Now we construct a new bipartite graph {${\cal G}_G$} by  replacing every  $a_i$ $(i\in [n])$ with a $d_i$-tuple $A_i=(a^i_1, \dots, a^i_{d_i})$ of vertices 
and every $b_i$ $(i\in [n'])$ with a $d'_i$-tuple $B_i=(b^i_1, \dots, b^i_{d'_i})$ of vertices. 
{If $G$ is clear from the context, we may omit the subscript and write ${\cal G}$ instead of ${\cal G}_G$ and $L$ instead of $L_G$.}

The partite sets of $\cal G$ are ${\cal A}=A_1\cup \dots \cup A_n$ and ${\cal B}=B_1\cup \dots \cup B_{n'}$.
The edge set $E({\cal G})$ of ${\cal G}$ is defined as follows:  $a^i_pb^j_q$ is an edge in ${\cal G}$ if $a_ib_j\in E$, $b_j$ is the $p$'th element of $\ell(a_i)$ and $a_i$ is the $q$'th element of $\ell(b_j).$ Note that the maximum degree of $\cal G$ is 1.
Thus, ${\cal G}=({\cal A} ,{\cal B};E({\cal G})).$ Figure \ref{fig:reduction_vertex_cover2} shows $\cal G$ for $(G,L)$ in Figure \ref{fig:reduction_vertex_cover1}.

A \emph{matching} in $\cal G$
is a subset $M$ of $E({\cal G})$ such that for every $i\in [n]$, 
there is at most one edge between $A_i$ and $\cal B$ that belongs to $M$,
and for every $i'\in [n']$, there is 
at most one edge between $\cal A$ and $B_{i'}$ that belongs to $M$.
A matching $M$ is said to be \emph{maximal} if no edge of $\cal G$ can be added to it.
A matching $M$ is \emph{stable} if there is no edge
of the form $a^i_j b^{i'}_{j'} \in E({\cal G}) \setminus M$ such that:
\begin{eqnarray*}
j &< &\min \{n+1\} \cup \{r \mid a^i_r \mbox{ is matched in } M\} \GG{\mbox{ and }}\\
j' & < &\min \{n'+1\} \cup 
                \{r' \mid b^{i'}_{r'} \mbox{ is matched in } M\}
\end{eqnarray*}
Otherwise, such an edge is called a \emph{blocking} edge. Figure \ref{fig:reduction_vertex_cover2} illustrates the notion of a blocking edge in $\cal G$.
{It is not hard to see that the blocking edges in $G$ (the stable matchings of $G$, respectively) are in a one-to-one correspondence to the blocking edges in $\cal G$ (the stable matchings of $\cal G$, respectively).} 

A \emph{swap} $\sigma$ is an adjacent transposition of two vertices in some $A_i$ $(i\in [n])$ or $B_{i'}$ $({i'}\in [n'])$. Thus, swaps can be 
of the form $a^i_j \leftrightarrow a^i_{j-1}$ in some $A_i$,
or $b^{i'}_{j'} \leftrightarrow b^{i'}_{j'-1}$ in some $B_{i'}$.
We denote by $\sigma(\mathcal{A}, \mathcal{B}) = 
(\mathcal{A}', \mathcal{B}')$ the resulting values 
of $\mathcal{A}$ and $\mathcal{B}$.  
\begin{figure}[hbtp]
    \centering
    \begin{subfigure}[b]{0.4\textwidth}
        \centering
        \tikzstyle{vertexDOT}=[scale=0.25,circle,draw,fill]
\tikzstyle{vertexY}=[circle,draw, top color=gray!5, bottom color=gray!30, minimum size=12pt, scale=0.85, inner sep=0.99pt]
\tikzstyle{vertexZ}=[circle,draw, top color=gray!5, bottom color=red!70, minimum size=7pt, scale=0.85, inner sep=0.5pt]
\begin{tikzpicture}[scale=0.55]

\node (x1) at (3,7) [vertexY] {$a_1$};
\node (x2) at (3,4) [vertexY] {$a_2$};
\node (x3) at (3,1) [vertexY] {$a_3$};

\node (y1) at (7,7) [vertexY] {$b_1$};
\node (y2) at (7,4) [vertexY] {$b_2$};
\node (y3) at (7,1) [vertexY] {$b_3$};

\draw [densely dashed, line width=0.04cm] (x1) -- (y1);
\draw [line width=0.08cm] (x1) -- (y2);
\draw [densely dashed, line width=0.04cm] (x1) -- (y3);

\draw [line width=0.08cm] (x2) -- (y1);
\draw [densely dashed, line width=0.04cm] (x2) -- (y2);
\draw [line width=0.02cm] (x2) -- (y3);

\draw [densely dashed, line width=0.04cm] (x3) -- (y1);
\draw [line width=0.02cm] (x3) -- (y2);

\node [scale=0.6] at (2.5,7.7) {$(1,3,2)$};
\node [scale=0.6] at (2.5,4.7) {$(2,1,3)$};
\node [scale=0.6] at (2.5,1.7) {$(1,2)$};

\node [scale=0.6] at (7.5,7.7) {$(1,3,2)$};
\node [scale=0.6] at (7.5,4.7) {$(2,1,3)$};
\node [scale=0.6] at (7.5,1.7) {$(2,1)$};

\end{tikzpicture}
        \caption{An example of an SMP instance with the preference lists.
            A given matching $M$ is represented by the bold edges and
            the blocking edges are dashed.}
        \label{fig:reduction_vertex_cover1}
    \end{subfigure}
\hspace{2cm}
    \begin{subfigure}[b]{0.4\textwidth}
        \centering
        \tikzstyle{vertexDOT}=[scale=0.25,circle,draw,fill]
\tikzstyle{vertexY}=[circle,draw, top color=gray!5, bottom color=gray!30, minimum size=12pt, scale=0.65, inner sep=0.99pt]
\tikzstyle{vertexZ}=[circle,draw, top color=gray!5, bottom color=red!70, minimum size=7pt, scale=0.85, inner sep=0.5pt]
\begin{tikzpicture}[scale=0.55]

\node (x11) at (3,10) [vertexY] {$a_1^1$};
\node (x12) at (3,9) [vertexY] {$a_2^1$};
\node (x13) at (3,8) [vertexY] {$a_3^1$};

\node (x21) at (3,6) [vertexY] {$a_1^2$};
\node (x22) at (3,5) [vertexY] {$a_2^2$};
\node (x23) at (3,4) [vertexY] {$a_3^2$};

\node (x31) at (3,2) [vertexY] {$a_1^3$};
\node (x32) at (3,1) [vertexY] {$a_2^3$};

\node [scale=0.8] at (2.2,10.8)  {$A_1$};

\node [scale=0.8] at (2.2,6.8)  {$A_2$};

\node [scale=0.8] at (2.2,2.8)  {$A_3$};

\draw [rounded corners] (2.5,0.5) rectangle (3.5,2.5);
\draw [rounded corners] (2.5,3.5) rectangle (3.5,6.5);
\draw [rounded corners] (2.5,7.5) rectangle (3.5,10.5);

\node (y11) at (8,10) [vertexY] {$b_1^1$};
\node (y12) at (8,9) [vertexY] {$b_2^1$};
\node (y13) at (8,8) [vertexY] {$b_3^1$};

\node (y21) at (8,6) [vertexY] {$b_1^2$};
\node (y22) at (8,5) [vertexY] {$b_2^2$};
\node (y23) at (8,4) [vertexY] {$b_3^2$};

\node (y31) at (8,2) [vertexY] {$b_1^3$};
\node (y32) at (8,1) [vertexY] {$b_2^3$};

\node [scale=0.8] at (8.7,10.8)  {$B_1$};

\node [scale=0.8] at (8.7,6.8)  {$B_2$};

\node [scale=0.8] at (8.7,2.8)  {$B_3$};


\draw [rounded corners] (7.5,0.5) rectangle (8.5,2.5);
\draw [rounded corners] (7.5,3.5) rectangle (8.5,6.5);
\draw [rounded corners] (7.5,7.5) rectangle (8.5,10.5);

\draw [densely dashed, line width=0.04cm] (x11) -- (y11);
\draw [densely dashed, line width=0.04cm] (x12) -- (y32);
\draw [line width=0.08cm] (x13) -- (y22);

\draw [densely dashed, line width=0.04cm] (x21) -- (y21);
\draw [line width=0.08cm] (x22) -- (y13);
\draw [line width=0.02cm] (x23) -- (y31);

\draw [densely dashed, line width=0.04cm] (x31) -- (y12);
\draw [line width=0.02cm] (x32) -- (y23);

\end{tikzpicture}
        \caption{The corresponding new representation, ${\cal G}_G$. The higher
            a vertex is placed, the higher its preference.}
        \label{fig:reduction_vertex_cover2}
    \end{subfigure}
    \caption{Illustration of the correspondence between the
        two representations of an SMP instance.}
    \label{fig:another_representation}
\end{figure}
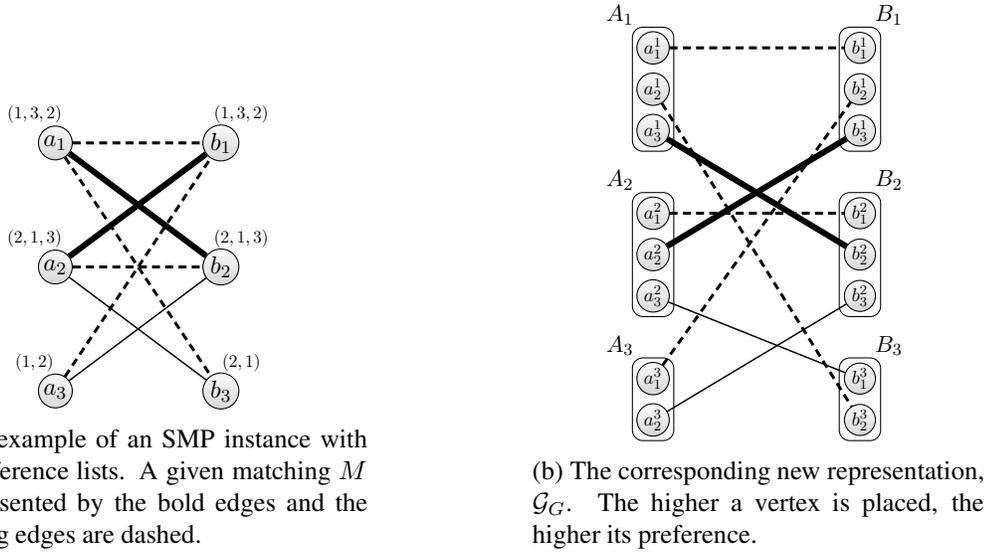

We also consider the more general capacitated extension of SMP (it is equivalent to many-to-many SMP), where every vertex of $G$ has a positive weight given by a capacity function $c : A \cup B \to \mathbb{Z}_{+}.$ Then we are looking for a stable generalization of a matching, which is a subgraph $S$ of $G$ where the degree $d_S(v)\le c(v)$ for every $v\in A\cup B.$ The notion of a blocking edge can be extended from the uncapacited case (where $c(v)=1$ for every $v$) to the capacitated case by saying that $ab\in E\setminus E(S)$ is a {\em blocking edge} of $S$
if the following two conditions hold simultaneously: (i) either $d_S(a) < c(a)$ or there is a $b'\in B$
such that $ab'\in E(S)$ and {$a$ prefers $b$ to $b'$
	and, similarly, (ii) either $d_S(b) < c(b)$ or there is an $a'\in A$ such that $a'b\in E(S)$ and {$b$ prefers $a$ to $a'.$} 
A subgraph $S$ is {\em stable} if it has no blocking edges.

\section{Making a given matching stable}\label{sec:generalisation_to_capacities}
\MW{In this section, we show that Problem~1 can be solved in polynomial
time, even in the capacitated case; i.e., given a capacitated SMP
instance $I=(G, L_G)$ and a subgraph $S$ of $G$, there is a
polynomial-time algorithm for computing the minimum number of swaps
required to make $S$ stable (if possible).  Our algorithm builds on
the extended bipartite graph $\cal G$ defined in the previous
section. In the uncapacitated case, it is possible to show that the
problem reduces to a vertex cover instance on a bipartite graph
derived from $\cal G$; however, an efficient solution for Problem~1 in
the uncapacitated case was already given by Bredereck et
al.~\cite{BredereckCKLN20}, so we omit the details and focus on the
capacitated case.}

\MW{Note that the standard method of reducing capacitated matching
problems to uncapacitated problems, by taking $c(v)$ copies of every
vertex $v$ (with capacity $c(v)$), does not work here since the cost
of making swaps in $\ell(v)$ is lost in the translation. 
Instead we solve the problem using more powerful tools.  We show that
the number of swaps required in $\ell(v)$, as a function of a set of
blocking edges to fix at $v$, is \emph{submodular}, allowing us to
reduce Problem~1 to a case of submodular function minimization.}
	For a set \(X\), a function $f\colon \mathcal{P}(X)\rightarrow \mathbb{N}$ is submodular if for every \(Y, Z\subseteq X\) we have that $f(Y)+f(Z)\ge f(Y\cup Z) + f(Y\cap Z)$. In our proof it will be convenient to consider also the following equivalent definition of a submodular function. It is well known that $f\colon \mathcal{P}(X)\rightarrow \mathbb{N}$ is submodular if and only if for
	\(Y\subseteq X\) and for every \(x_1, x_2\in  X\setminus Y\) such that $x_1\neq x_2$ it holds that \(f(Y\cup \{x_1,x_2\}) - f(Y\cup \{x_1\})\le f(Y\cup \{x_2\}) - f(Y)\). Before we proceed to proof our main result of this section, we state the following result in \cite{schrijver2000combinatorial}. 
	
	\begin{theorem}[\cite{schrijver2000combinatorial}]\label{thm:schrijver}
		Let $X$ be a finite set and $f\colon \mathcal{P}(X)\rightarrow \mathbb{N}$ a submodular function such that for every \(Y\subseteq X\) the value \(f(Y)\) can be computed in polynomial time. Then there is a strongly polynomial-time algorithm minimizing the function \(f\).
	\end{theorem}

	\begin{theorem}\label{thm:capacitated_algorithm}
		Let $(G=(A,B,E, c), L_G)$ be an SMP instance
		with capacities and $S$ a subgraph of $G$ such that $d_S(v)\le c(v)$ for every $v\in A\cup B$.
		There exists a polynomial-time algorithm that 
		finds a minimum length sequence $\sigma_1, \dots, \sigma_k$ of swaps
		such that $S$ is a stable subgraph in
		$(G, (\sigma_k\circ\dots\circ\sigma_1(L_G)))$.
	\end{theorem}
	\begin{proof}
		
		Recall that $n=|A|$, $n'=|B|$, $A= \{a_1, \ldots, a_n\}$, and $B= \{b_1, \ldots, b_{n'}\}$. We first construct the extended graph \(\mathcal{G}_G\) with $\cA=A_1\cup \ldots \cup A_{n}$ and  \(\cB=B_1\cup \ldots \cup B_{n'}\). Recall that $A_i=\{a^i_1,\ldots, a^i_{d_G(a_i)}\}$ corresponds to the vertex $a_i\in A$ and  $B_j=\{b^j_1,\ldots, b^j_{d_G(b_j)}\}$ corresponds to the vertex $b_j\in B$.
		Note that $S$ induces a matching in \(\mathcal{G}_G\) and we denote this matching $M$.
		
		Our goal is to determine a minimal sequence of swaps making $M$ stable.
		Let $E_b$ be the set of blocking edges of $M$. Furthermore, let $E^A_b\subseteq E_b$
		be the subset of blocking edges of $M$ such that for every edge $a^i_pb^j_q\in E^A_b$
		it holds that $d_S(a_i) = c(a_i)$. Note that if $e=a^i_pb^j_q \in E_b \setminus E^A_b$, then $e$
		cannot become
		non-blocking by permuting $A_i$. Similarly let $E^B_b\subseteq E_b$
		be the subset of blocking edges for $M$ such that for every edge $a^i_pb^j_q\in E^B_b$
		it holds that $d_S(b_j) = c(b_j)$. Note that if $E^A_b\cup E^B_b\neq E_b$, then there is 
		a blocking edge in $E_b$ such that both endpoints of the edge do not have full capacity
		and $S$ cannot be made stable. 
		
		
		The idea of the proof is that making $M$ stable means
		unblocking every blocking edge, where unblocking a blocking edge corresponds to
		moving at least one of its endpoints below the vertices matched by $M$.
			Define a cost function $h: \mathcal{P}(E^A_b) \to \mathbb{N}$
		that maps a set $F \subseteq E_b$ of blocking edge
		to the minimum number of swaps to fix every edge in $F$ by moving its endpoint in $\cA$ and every other blocking edge by moving its endpoint in $\cB$.
		

		We define for every \(i\in \{1,2,\ldots, n\} \) 
		the function $f_i\colon \mathcal{P}(E_b) \to \mathbb{N}$
		such that 
		for every \(F\subseteq E_b\), the value \(f_i(F)\) is equal to the minimum
		number of swaps in $A_i$ moving every vertex in \(A_i\cap\{u \mid uv\in F\}\) below 
		the matched vertices, i.e., the vertices in \(A_i\cap \bigcup_{e\in M}e\). Furthermore,
		for every \(j\in \{1,2,\ldots, n'\} \) we define the function
		$g_j\colon \mathcal{P}(E_b) \to \mathbb{N}$\ such that 
		for every \(F\subseteq E_b\), the value \(g_i(F)\) is equal to the minimum
		number of swaps in $B_j$ moving every vertex in \(B_j\cap\{u \mid uv\in F\}\) below 
		the matched vertices, i.e., the vertices in \(B_j\cap \bigcup_{e\in M}e\).
		
		
		Now, let us consider a minimal sequence of swaps  $\sigma_1, \dots, \sigma_k$ making
		$M$ stable, and let $F$ be the subset of 
		blocking edges that are fixed using swaps in $\cA$; i.e., if $a^i_pb^j_q\in F$, then $d_S(a_i)=c(a_i)$ and \(a^i_p\) is below every matched vertex in $A_i$ after applying the sequence of swaps $\sigma_1, \dots, \sigma_k$ and if $a^i_pb^j_q\in E_b\setminus F$, then $d_S(b_j)=c(b_j)$ and \(b^j_q\) is below every matched vertex in $B_j$ after applying the sequence of swaps $\sigma_1, \dots, \sigma_k$. Clearly, $E_b\setminus E^B_b \subseteq F\subseteq E^A_b$.
		It is easy to see that the number of swaps in the sequence  $\sigma_1, \dots, \sigma_k$ that are in $A_i$ is at least $f_i(F)$ and the number of swaps that are in $B_j$ is at least  $g_{j}(E_b \setminus F)$.
		Therefore, the total number of swaps $k$ is at least 
		$h(F) = \sum_{i\in [n]} f_i(F) + \sum_{j\in [n']} g_{j}(E_b \setminus F)$. On the other 
		hand, it is easy to see that for any $F'$ such that \(E_b\setminus E^B_b \subseteq F'\subseteq E^A_b\) it is always possible to
		make $M$ stable with at most $h(F')$ many swaps.  Hence $k=h(F)$ and to find 
		a minimum length sequence of swaps, it suffices to find $F$ such that  $E_b\setminus E^B_b \subseteq F\subseteq E^A_b$ minimizing $h(F)$.
		To avoid the constraint $E_b\setminus E^B_b \subseteq F$, let $J'\subseteq [n']$ be the set of indices such that $j \in J'$ if and only if $d_S(b_j) < c(b_j)$. For every $j\in J'$ we define a function $g'_j\colon \mathcal{P}(E_b)\rightarrow \mathbb{N}\cup \{\infty\}$ such that for every $F\subseteq E_b$, we have $g'_j(F)=0$ if \(B_j\cap\{u \mid uv\in F\} = \emptyset\) and $g'_j(F)=\infty$\footnote{Here infinity can be replaced by some large number, e.g., $(nn')^2$ would suffice.} otherwise. We define $h'\colon \mathcal{P}(E^A_b)\rightarrow \mathbb{N}\cup \{\infty\}$ such that $h'(F)=h(F)+\sum_{j\in J'} g'_{j}(E_b \setminus F)$. Note that for every $F$ such that \(E_b\setminus E^B_b \subseteq F\subseteq E^A_b\), we have $h'(F)=h(F)$. On the other hand if $a^i_pb^j_q\in (E_b\setminus E^B_b)\setminus F$, then  $a^i_pb^j_q\in E_b\setminus F$ and $g'_j(F)=\infty$ implying $h'(F)=\infty$. 
		
		We will show that $h'$ is submodular and for fixed \(F\subseteq E_b\) we can compute $h'(F)$ in polynomial time. Then we can use Theorem~\ref{thm:schrijver}
		to minimise $h'$ in polynomial time.
		
		Elementary operations on submodular functions show that it suffices to
		prove that every function $f_i$ as well as every function $g_j$ and every function $g'_j$ are submodular.
		It is easy to see that $g'_j$ is submodular and computable in polynomial time. 
		The proof for a function $g_j$ is analogous to the proof of 
		submodularity of function $f_i$ and hence it is omitted. Let us now 
		fix \(i\in [n]\) for the rest of the proof. 

		For the ease of explanation, we partition $A_i$ into three kinds
		of vertices:
		\begin{enumerate}[label=(\roman*)]
			\item the \emph{red} vertices -- the endpoint of
			a blocking edge in $F$,
			\item the \emph{blue} vertices -- the matched vertices in $M$, and
			\item the \emph{black} vertices -- the rest.
		\end{enumerate}

		First, we make some basic observations:
		\begin{itemize}
			\item red vertices never need to go up,
			\item blue vertices never need to go down,
			\item two vertices of the same kind never need to be swapped,
			\item if $a^i_r$ is red and $a^i_b$ blue below $a^i_r$, then
			$a^i_r$ and $a^i_b$ must be swapped.
		\end{itemize}
		
		For any vertex $a^i_j$, we denote by
		$B_{F}(j)$ the number of blue vertices below $a^i_j$
		and $R_{F}(j)$ the number of red vertices above $a^i_j$.

		The key observation is the following claim.
		
		\begin{claim}\label{clm:value_sv}
			$f_i(F) = \sum_{a^i_r \text{ is red}} B_{F}(r) + \sum_{a^i_j \text{ is black}}
			\min\{B_{F}(j),R_{F}(j)\}.$
		\end{claim}
		
		\begin{proof}
Let $s(j)$ be the number of swaps with the vertex $a^i_j$. 
Observe that $$f_i(F) = \sum_{a^i_j \in A_i, a^i_j \text{ is black}} s(j) + N,$$
where $N$ is the number of swaps between a red and a blue vertex,
but $$N = |\{(a^i_r,a^i_b) \mid a^i_r \text{ is red, }a^i_b \text{ is blue, and }
a^i_r \text{ is above } a^i_b\}|,$$ which is equal to $$\sum_{a^i_r \text{ is red}} B_{F}(r)= \sum_{a^i_b \text{ is blue}}R_{F}(b)$$ and can be computed in time $\bigO(n^2)$. Note that given \(F\) the value $N$ is fixed. Hence, it suffices to minimize $\sum_{a^i_j \text{ black}} s(j)$.
			
		This is because our goal is to find the minimum length sequence of swaps 
				such that all red vertices are below the blue vertices. 
				First, it is easy to see that $s(j) \ge \min\{B_{F}(j),R_{F}(j)\}$,  because 
				we have either swap with \(a^i_j\) all the blue vertices below $a^i_j$, so they end up above $a^i_j$
				or all the red vertices that are above $a^i_j$, so they end up below $a^i_j$.
				
					It remains to give a sequence of swaps such that
			$s(j) = \min\{B_{F}(j),R_{F}(j)\}$. One can verify that the following procedure gives such a sequence:
			while at least one of the following rules can be applied, perform it
			\begin{itemize}
				\item if $a^i_r$ is a red vertex, $a^i_b$ blue vertex such that $r=b-1$, that is $a^i_r$ is a predecessor of $a^i_b$, 
				swap $a^i_r$ and $a^i_b$,
				\item if $a^i_j$ is black, $a^i_{j+1}$ is blue, and $B_{F}(j) \leq R_{F}(j)$, then swap $a^i_j$ and $a^i_{j+1}$,
				\item if $a^i_j$ is black, $a^i_{j-1}$ red, and $R_{F}(j) < B_{F}(j)$, then swap $a^i_j$ and $a^i_{j-1}$.
			\end{itemize}
		
				Let us fix a black vertex $a^i_j$. It is easy to show by induction on $\min\{B_{F}(j),R_{F}(j)\}$ that the number of swaps 
				of this vertex is exactly $\min\{B_{F}(j),R_{F}(j)\}$. Clearly if $\min\{B_{F}(j),R_{F}(j)\} = 0$, then $a^i_j$ never swaps with any vertex in the above sequence. If $R_{F}(j) < B_{F}(j)$, then $a^i_j$ might swap with red vertex at position $a^i_{j-1}$. In this case the black vertex we consider becomes $a^i_{j-1}$ and the number of blue vertices below remains the same and the number of red vertices above decreases by one and $s(j) = \min\{B_{F}(j),R_{F}(j)\}$ by induction. Similarly, if $R_{F}(j) \ge B_{F}(j)$, $a^i_j$ might swap with blue vertex at position $a^i_{j+1}$ and $s(j) = \min\{B_{F}(j),R_{F}(j)\}$ follows by induction as well.
				It remains to show that if none of these rules can be applied, then
				there is no red vertex above a blue one.
				For the sake of  contradiction let us assume that none of the rules apply and there is some red vertex above a blue vertex. 
				Let $a^i_r$ be a red vertex and $a^i_b$ a blue vertex such that $r<b$ (i.e., $a^i_r$ is above $a^i_b$) and all vertices between $a^i_r$ be a red vertex and $a^i_b$ are black. Clearly $a^i_{r}\neq a^i_{b-1}$, else the first rule applies. Hence both vertices $a^i_{r+1}$ and $a^i_{b-1}$ are black (possibly $r+1=b-1$). Now, all vertices between $a^i_{r+1}$ and $a^i_{b-1}$ are black and it follows that $B_{F}(r+1)=B_{F}(b-1)$ and $R_{F}(r+1)=R_{F}(b-1)$. Since the second rule does not apply, it follows that $B_{F}(b-1) > R_{F}(b-1)$. However that implies $R_{F}(r+1) < B_{F}(r+1)$ and the third rule applies, which is a contradiction.
				%
				%
				%
				%
				%
		\end{proof}	
		
		Let  $N_{F'}(j)$ be the number of black
vertices $a_b^i$ below $a_j^i$ such
that $R_{F'}(b) <B_{F'}(b)$. From Claim~\ref{clm:value_sv} we can now deduce that
			if $e,e' \in E_b$ such that $e'\cap A_i = a^i_j$, then
			\begin{align*}
				f_i(F&+e+e') - f_i(F+e)  \\ 
				&= B_{F+e+e'}(j) - \min\{B_{F+e}(j), R_{F+e}(j)\} + N_{F+e}(j)\\
				&= B_{F+e}(j) - \min\{B_{F+e}(j), R_{F+e}(j)\} + N_{F+e}(j) \\
				&= \max\{0, B_{F+e}(j) - R_{F+e}(j)\} + N_{F+e}(j) \\
				&\leq \max\{0, B_{F}(j) - R_{F}(j)\} + N_F(j)\\
				&=f_i(F+e') - f_i(F).
			\end{align*}
			The above inequality follows because $R_{F+e}(j) \geq R_{F}(j)$ and $B_{F+e}(j)=B_{F}(j)$. Furthermore $N_{F+e}(j) \leq N_F(j)$ since $R_F(b)
\leq R_{F+e}(b)$ and $B_{F+e}(b)=B_F(b)$.
		It follows that $f$ is computable in polynomial time and submodular.
		We can now efficiently minimise the total cost $h'$ 
		by the algorithm of Theorem~\ref{thm:schrijver}
		and thus compute in polynomial time a minimal sequence of swaps
		making $M$ and hence $S$ stable.
	\end{proof}

	\section{Hardness results}\label{sec:NP_hardness}
	
	\newcommand{\smMinSwaps}{\textsc{Swap Distance to Perfect Stable Matching}}
	\newcommand{\psmd}{$\Delta$-PSM}
	In this section we consider the problem of finding a minimum length sequence of swaps that leads to an arbitrary maximum stable matching. Note that  even though we can find minimum length sequence for a fixed maximum matching, there might be exponentially many different maximum matchings, so the algorithms from the previous section cannot be applied to solve this problem. Indeed, we will show that the problem of deciding whether there exists a sequence of at most $k$ swaps that leads to an instance with a perfect stable matching is NP-hard and W[1]-hard parameterized by $k$. We will call this problem \smMinSwaps\ (\psmd).  
	\EE{Note that Chen et al. \cite{ChenSS19} proved that \psmd\ is W[1]-hard if $G$ has a perfect matching and the parameter is $n_u,$ the number of vertices unmatched in any stable matching of $G$. The following two propositions show that W[1]-hardness with parameter $k$ is strictly stronger than the result by Chen et al. \cite{ChenSS19}.}
	
	\begin{proposition}\label{prop:comparison_to_chen_at_all}
		\EE{Let $I=(G, L)$ be an SMP instance, $k$ the minimum number of swaps that leads to an instance with a perfect stable matching and $n_u$ the number of vertices unmatched in any stable matching of $G$, then $n_u\le 2k$.}
	\end{proposition}
	
	\begin{proof}
\EE{Let $M$ be a stable matching in $I$ and $M'$ a perfect stable matching in an instance $I'$ obtained from $I$ by exactly $k$ swaps. Note that $M\cup M'$ induces a union of cycles and paths. Moreover, $M\cup M'$ has to contain at least $|M'| - |M| = \frac{n_u}{2}$ paths that start and end with an edge in $M'$ and the consecutive edges alternate between an edge in $M$ and an edge in $M'$. Note that any such alternating path starts and ends in a vertex unmatched by $M$. Moreover, the set of unmatched vertices is the same in any stable matching of $I$ due to the Rural Hospitals Theorem. Let $P=v_1v_2v_3\ldots v_{q}$ be one such alternating path, where $v_1$ and $v_q$ are unmatched by $M$. The edges} \GG{$v_{2i-1}v_{2i}$, for $i\in [\frac{q}{2}]$, are in $M'$ and the edges $v_{2i}v_{ 2i+1}$, for $i\in [\frac{q-2}{2}]$, are in $M$.} 
\EE{We show that there is a swap in a preference list of at least one vertex of $P$. The proposition then follows from the fact that there are at least $\frac{n_u}{2}$ such paths and these paths are vertex disjoint.} 
		
		\EE{For the sake of contradiction, let us assume that there is no swap in a preference list of any vertex on $P$. First note that since $v_1$ is unmatched by $M$, the vertex $v_2$ prefers $v_3$ to $v_1$, otherwise $v_1v_2$ would be a blocking edge of $M$. We now show by induction that since there is no swap in a preference list of any vertex of $P$, then $v_2$ also prefers $v_1$ to $v_3$, which is a contradiction. Since, $v_q$ is unmatched by $M$ and $v_{q-2}v_{q-1}\in M$, the vertex $v_{q-1}$ prefers $v_{q-2}$ to $v_{q}$. As an induction hypothesis assume that for some $i$, $2 < i < q$, the vertex $v_i$ prefers $v_{i-1}$ to $v_{i+1}$. Edges $v_{i-2}v_{i-1}$ and $v_{i}v_{i+1}$ are together in a stable matching. Since by the induction hypothesis $v_i$ prefers $v_{i-1}$ to $v_{i+1}$, if the vertex $v_{i-1}$ prefers $v_{i}$ to $v_{i-2}$, then the edge $v_{i-1}v_i$ is blocking for this stable matching. We conclude that $v_{i-1}$ prefers $v_{i-2}$ to $v_{i}$ and the induction step holds. Applying repeatedly the induction step, we obtain that $v_2$ prefers $v_1$ to $v_3$, which is the desired contradiction.} 
	\end{proof}


\AY{
\begin{proposition}\label{prop:comparison_to_chen_at_all2}
For every $n\in\mathbb{N}$, $n\ge 1$, there exists an SMP instance $(G, L)$ such that $G$ is a bipartite graph with $2n+2$ vertices that admits a perfect matching, only two vertices are unmatched in the unique stable matching of $G$ and the minimum number of swaps that leads to an instance with a perfect stable matching is at least $n$.
\end{proposition}
\begin{proof}
Let $G$ be a path of length $2n+1.$ That is, $V(G)=\{v_0,v_1,\ldots, v_{2n+1}\}$ and $E(G)=\{v_0 v_1, v_1 v_2, v_2 v_3, \ldots, v_{2n} v_{2n+1}\}.$  Assume that for all $i\in [n]$, 
$v_{2i}$ prefers $v_{2i-1}$ over $v_{2i+1}$ and  $v_{2i-1}$ prefers $v_{2i}$ over $v_{2i-2}$. 
  Now $M=\{v_1 v_2, v_3 v_4, \ldots, v_{2n-1} v_{2n}\}$ is the only stable matching in $G$.   If we want $M'=\{v_0 v_1, v_2 v_3, \ldots, v_{2n} v_{2n+1}\}$ to become a stable matching then we need to change the preference of $v_{2i-1}$ or $v_{2i}$ for every  $i\in [n]$, as otherwise $v_{2i-1}v_{2i}$ would be a blocking edge.
Hence, $k \geq n.$
\end{proof}
}
	
Gupta et al.	\cite{gupta2020parameterized} studied a related problem called \textsc{Almost Stable Marriage (ASM)}, which takes as an input an instance $(G,L)$ of SMP and two non-negative integers $k$ and $t$ and asks whether there is a matching whose size is at least $t$ more than the size of a stable matching in $G$ such that the matching has at most $k$ blocking edges. The authors show that ASM is W[1]-hard parameterized by $k+t$ even when the input graph $G$ has max degree $3$. While they did not state it explicitly their hardness proof implies the following result.
	
	\begin{theorem}[\cite{gupta2020parameterized}]\label{thm:ASM_hardness}
		ASM is NP-hard and W[1]-hard parametrized by $k+t$ even on instances $\mathcal{I}=(G,L,k,t)$ such that 
		\begin{enumerate}
			\item $G$ has max degree three,
			\item $G$ admits a perfect matching and a stable matching of size $\frac{|V(G)|}{2}-t$, and
			\item If $\mathcal{I}$ is a YES-instance, then in every perfect matching $M$ with at most $k$ blocking edges, every blocking edge is incident to a vertex of degree two.
		\end{enumerate}
		Moreover, ASM does not admit an $|V(G)|^{o(\sqrt{k+t})}$ algorithm, unless ETH fails.  
	\end{theorem}


	
	To show that \psmd\ is NP-hard and W[1]-hard parameterized by $k$, we start with the following observation.
	
	\begin{observation}\label{obs:decrease_after_1_swap}
		Let $I=(G=(A,B;E), L)$ be an SMP instance, $M$ a matching such that there are exactly $k$ blocking edges of $M$ in $E$, and let $\sigma$ be a swap. Then there are at least $k-1$ blocking edges of $M$ in $E$ in \(\sigma(I)\).
	\end{observation}
\begin{proof}
			Swap \(\sigma\) increases preference of exactly one vertex $v$ in exactly one list of some vertex $u$. Hence the only blocking edge in the original instance that can become non-blocking is the edge $uv$. 
		\end{proof}
		
Now we are ready to show the hardness of \psmd.
	
	\begin{theorem}\label{thm:hardness_swap_distance}
		\psmd\ is NP-complete. Moreover, \psmd\ is W[1]-hard parameterized by the length $k$ of the sought sequence of swaps.
	\end{theorem}
	\begin{proof}
Given a sequence of $k$ swaps, it is easy to verify that the instance obtained by applying the sequence of swaps admits a stable perfect matching. 
			Moreover, given a perfect matching $M$ of $G$, \GG{a polynomial in $n$ many swaps suffices to make $M$ stable}. Hence \psmd\ is in NP. 
			To show that the problem is NP-hard and W[1]-hard parameterized by $k$, we reduce from ASM. 
		Let $\mathcal{I}=(G,L,k,t)$ be an instance of ASM such that 
			\begin{enumerate}
				\item $G$ has max degree three,
				\item $G$ admits a perfect matching and a stable matching of size $\frac{|V(G)|}{2}-t$, and
				\item If $\mathcal{I}$ is a YES-instance, then in every perfect matching $M$ with at most $k$ blocking edges, every blocking edge is incident to a vertex of degree two.
			\end{enumerate}
			
			Recall that ASM remains NP-hard and W[1]-hard parameterized by $k+t$ on such instances by Theorem~\ref{thm:ASM_hardness}.   
		We let $\mathcal{I}'= (G,L,k)$ be an instance of \psmd\ and we will show that \(\mathcal{I}\) is a YES-instance of ASM if and only if \(\mathcal{I}'\) is a YES-instance of \psmd{}. Since, this construction can be performed in polynomial time (we just take the same instance) and the reduction is parameter-preserving, this implies that \psmd\ is NP-hard and W[1]-hard parameterized by $k$. 
		
		Observation~\ref{obs:decrease_after_1_swap} implies that if a graph $G$ does not admit a perfect matching $M$ with at most $k$ blocking edges, then $\cI'$ is a NO-instance of \psmd. As the instance \(\cI\) of ASM admits a perfect matching and a stable matching of size $\frac{|V(G)|}{2}-t$, it is easy to see that if \(\cI\) is a NO-instance of ASM, then $G$ does not admit a perfect matching $M$ with at most $k$ blocking edges and $\cI'$ is a NO-instance of \psmd.
		On the other hand, if \(\cI\) is a YES-instance of ASM, then  $G$ admits a perfect matching $M$ with at most $k$ blocking edges and, by propery 3 of \(\cI\), every blocking edge of $M$ is incident to a vertex of degree $2$. It is easy to see that if $e=uv$ is a blocking edge and $u$ has degree $2$, then it suffices to transpose the two vertices in \(l(u)\). This swap also does not introduce new blocking edges, as after this swap, vertex $u$ is matched with its most preferred choice. Therefore, in this case, there is a sequence of at most $k$ swaps that makes $M$ a stable perfect matching and \(\cI'\) is a YES-instance. 
	\end{proof}
	
	The ETH lower bound for ASM implies that \psmd\ does not admit an algorithm of running time of $(n+n')^{o(\sqrt{k})}$, where $k$ is the length of the sought solution. On the other hand, at each step there are only $2nn'$ possible swaps and so there is a simple $O((2nn')^{k})$ algorithm enumerating all possible sequences of at most $k$ swaps. By a small change to the reduction by Gupta et al. \cite{gupta2020parameterized} one can show that a significant improvement over the $(2nn')^{k}$ algorithm is unlikely. 
	
	In what follows we will sketch how to adapt the hardness proof by Gupta et al. \cite{gupta2020parameterized} to show that ASM, and hence also \psmd, does not admit an $(n+n')^{o(k/\log k)}$-time algorithm, unless ETH fails, and even existence of an $(n+n')^{o(k)}$-time algorithm would imply a breakthrough in the area of parameterized algorithms and complexity.
	
	The hardness proof of Gupta et al. \cite{gupta2020parameterized} is a reduction from \textsc{Multicolored Clique} (MCQ), where we are given a graph $G$ and a partition of $V(G)$ into $q$ parts $V_1,\ldots,V_q$; the goal is to decide the existence of a set $S\subseteq V(G)$ such that $|V_i\cap S|=1$, for all $i\in [q]$, and $G[S]$ induces a clique, that is, there is an edge between every pair of vertices in $G[S]$. It is well known that MCQ does not admit an $|V(G)|^{o(q)}$ algorithm, unless ETH fails. 
	The main idea of the reduction 
		is to introduce three types of gadgets - ``vertex set" gadgets, ``edge set" gadgets, and ``connection" gadgets. They introduce one ``vertex set" gadget for every vertex set $V_i$, one ``edge set" gadget for every $1\le i < j \le q$ that represents the set of edges between $V_i$ and $V_j$. Finally the "connection" gadgets connect a ``vertex set" gadget for $V_i$ with every "edge set" gadget for edges between $V_i$ and $V_j$, $j\neq i$. The parameter $t$ is then the number of ``vertex set" gadgets + the number of ``edge set" gadgets and $k$ is the number of ``vertex set" gadgets + 2 times the number of ``edge set" gadgets. Hence $t=q+\binom{q}{2}$ and $k=q^2$. Moreover, every perfect matching of ``vertex set" gadget forces at least one blocking edge inside the gadget that depends on the selected vertex for the clique in the corresponding vertex set. Every perfect matching of ``edge set" gadget forces at least two blocking edges inside the gadget that depend on the selected edge for the clique in the corresponding vertex set. Finally, a ``connection" gadget is just a set of edges between a ``vertex set" gadget and an ``edge set" gadget that contains blocking edge if the edge selected by the ``edge set" gadget is not incident to the vertex selected by the ``vertex set" gadget.
		
		Given this high-level description of  the hardness proof of \cite{gupta2020parameterized} for ASM parameterized by $k+t$, we sketch how one can adapt it to obtain $(n+n')^{o(k/\log k)}$ lower bound, under ETH. Namely, reduction will be analogous, but from a different problem. 
	In \textsc{Partitioned Subgraph Isomorphism} (PSI) we are given on the input 
	{two undirected graphs $G$ and $H$ with $|V(H)| \le |V(G)|$ ($H$ is \emph{smaller}) and a mapping $\psi\colon V(G) \to V(H)$ and the task is to determine whether
		{$H$ is isomorphic to a subgraph of $G$ (i.e., is there an injective mapping $\phi\colon V(H) \to V(G)$ such that $\{\phi(u),\phi(v)\} \in E(G)$ for each $\{u,v\} \in E(H)$ and $\psi \circ\phi$ is the identity)?} Observe that MCQ is PSI, where the smaller graph $H$ is a complete graph.
		
		\begin{theorem}[see {\cite{Marx10}} and \cite{EibenKPS19}]\label{cor:psi_hard}
			If there is an algorithm $\mathbb{A}$ and an arbitrary
			function $f$ such that $\mathbb{A}$ correctly decides every instance $(G,H)$
			of PSI with the smaller graph $H$ being 3-regular in time $f(|V(H)|)|V(G)|^{o(|V(H)|/\log |V(H)|)}$, then ETH fails.\footnote{We would like to point out that, as far as we know, it is open whether PSI admits even $f(|V(H)|)|V(G)|^{o(|V(H)|)}$ algorithm.}
		\end{theorem}
		
		Note that the mapping $\psi\colon V(G) \to V(H)$ partitions the vertices of $V(G)$ into $q=|V(H)|$ many parts $V_1,\ldots, V_q$, each corresponding to a specific vertex of $H$. Moreover, we wish to select in each part $V_i$, $i\in [q]$, exactly one vertex $v_i$, such that if $uw\in E(H)$ and $V_i$ corresponds to $u$ and $V_j$ corresponds to $w$, then $v_iv_j$ is an edge in $G$.} The reduction from PSI to ASM is precisely the same as the reduction from MCQ to ASM, with only difference that we have ``edge set" gadgets only for pairs $1\le i < j \le q$ such that the sets $V_i$ and $V_j$ correspond to adjacent vertices of $H$. If $H$ is $3$-regular, then the number of ``edge set" gadgets is $\frac{3q}{2}$ and hence in the instance of ASM we obtain in the reduction, we get $t=\frac{5q}{2}$ and $k=4q$, implying that ASM does not admit $|V(G)|^{o((k+t)/\log(k+t))}$-time algorithm, unless ETH fails. Following the proof of Theorem~\ref{thm:hardness_swap_distance}, we then obtain the following result.
		
		\begin{theorem}\label{thm:ETH_lower_bound}
			If there is an algorithm that for every instance of SMP with a bipartite graph on $n+n'$ vertices decides whether there is a sequence of at most $k$ swaps that results in an instance of SMP with a perfect matching in time $(n+n')^{o(k/\log k)}$, then ETH fails.
		\end{theorem}

\vspace{2mm}
{\bf Acknowledgement.} We are grateful to the referee for a number of helpful suggestions.


\begin{thebibliography}{10}

\bibitem{BalinskiR97}
M.~Balinski and G.~Ratier.
\newblock Of stable marriages and graphs, and strategy and polytopes.
\newblock {\em {SIAM} Rev.}, 39(4):575--604, 1997.

\bibitem{BalinskiR98}
M.~Balinski and G.~Ratier.
\newblock Graphs and marriages.
\newblock {\em American Mathematical Monthly}, 105(5):430--445, 1998.

\bibitem{BoehmerBHN20}
N.~Boehmer, R.~Bredereck, K.~Heeger, and R.~Niedermeier.
\newblock Bribery and control in stable marriage.
\newblock {\em Journal of Artificial Intelligence Research}, 71:993--1048,
  2021.

\bibitem{BredereckCKLN20}
R.~Bredereck, J.~Chen, D.~Knop, J.~Luo, and R.~Niedermeier.
\newblock Adapting stable matchings to evolving preferences.
\newblock In {\em The Thirty-Fourth {AAAI} Conference on Artificial
  Intelligence, {AAAI} 2020, The Thirty-Second Innovative Applications of
  Artificial Intelligence Conference, {IAAI} 2020, The Tenth {AAAI} Symposium
  on Educational Advances in Artificial Intelligence, {EAAI} 2020, New York,
  NY, USA, February 7-12, 2020}, pages 1830--1837. {AAAI} Press, 2020.

\bibitem{ChenSS19}
J.~Chen, P.~Skowron, and M.~Sorge.
\newblock Matchings under preferences: Strength of stability and trade-offs.
\newblock {\em ACM Transactions on Economics and Computation}, 9(4):1--55,
  2021.

\bibitem{CyganFKLMPPS15}
M.~Cygan, F.~V. Fomin, L.~Kowalik, D.~Lokshtanov, D.~Marx, M.~Pilipczuk,
  M.~Pilipczuk, and S.~Saurabh.
\newblock {\em Parameterized Algorithms}.
\newblock Springer, 2015.

\bibitem{EibenKPS19}
E.~Eiben, D.~Knop, F.~Panolan, and O.~Such{\'{y}}.
\newblock Complexity of the steiner network problem with respect to the number
  of terminals.
\newblock In R.~Niedermeier and C.~Paul, editors, {\em 36th International
  Symposium on Theoretical Aspects of Computer Science, {STACS} 2019, March
  13-16, 2019, Berlin, Germany}, volume 126 of {\em LIPIcs}, pages 25:1--25:17.
  Schloss Dagstuhl - Leibniz-Zentrum f{\"{u}}r Informatik, 2019.

\bibitem{gale1962college}
D.~Gale and L.~S. Shapley.
\newblock College admissions and the stability of marriage.
\newblock {\em The American Mathematical Monthly}, 69(1):9--15, 1962.

\bibitem{gupta2020parameterized}
S.~Gupta, P.~Jain, S.~Roy, S.~Saurabh, and M.~Zehavi.
\newblock On the (parameterized) complexity of almost stable marriage.
\newblock In N.~Saxena and S.~Simon, editors, {\em 40th {IARCS} Annual
  Conference on Foundations of Software Technology and Theoretical Computer
  Science, {FSTTCS} 2020, December 14-18, 2020, {BITS} Pilani, {K} {K} Birla
  Goa Campus, Goa, India (Virtual Conference)}, volume 182 of {\em LIPIcs},
  pages 24:1--24:17, 2020.

\bibitem{GutinNY21}
G.~Z. Gutin, P.~R. Neary, and A.~Yeo.
\newblock Unique stable matchings.
\newblock {\em CoRR}, abs/2106.12977, 2021.

\bibitem{ImpagliazzoP99}
R.~Impagliazzo and R.~Paturi.
\newblock Complexity of $k$-{SAT}.
\newblock In {\em Proceedings of the 14th Annual {IEEE} Conference on
  Computational Complexity, Atlanta, Georgia, USA, May 4-6, 1999}, pages
  237--240. {IEEE} Computer Society, 1999.

\bibitem{KamadaKojima:2015}
Y.~Kamada and F.~Kojima.
\newblock Efficient matching under distributional constraints: Theory and
  applications.
\newblock {\em The American Economic Review}, 105(1):67--99, 2015.

\bibitem{Maffray1992}
F.~Maffray.
\newblock Kernels in perfect line-graphs.
\newblock {\em J. Comb. Theory Ser. B}, 55(1):1?8, 1992.

\bibitem{MaiV18}
T.~Mai and V.~V. Vazirani.
\newblock Finding stable matchings that are robust to errors in the input.
\newblock In Y.~Azar, H.~Bast, and G.~Herman, editors, {\em 26th Annual
  European Symposium on Algorithms, {ESA} 2018, August 20-22, 2018, Helsinki,
  Finland}, volume 112 of {\em LIPIcs}, pages 60:1--60:11, 2018.

\bibitem{MaiV2020}
T.~Mai and V.~V. Vazirani.
\newblock An efficient algorithm for fully robust stable matchings via join
  semi-sublattices.
\newblock {\em CoRR}, abs/1804.05537v4, 2020.

\bibitem{Marx10}
D.~Marx.
\newblock Can you beat treewidth?
\newblock {\em Theory Comput.}, 6(1):85--112, 2010.

\bibitem{McVitieW70}
D.~McVitie and L.~Wilson.
\newblock Stable marriage assignment for unequal sets.
\newblock {\em BIT Numerical Math.}, 10(3):295--309, 1970.

\bibitem{Roth84}
A.~E. Roth.
\newblock The evolution of the labor market for medical interns and residents:
  a case study in game theory.
\newblock {\em Journal of Political Economy}, 92(6):991--1016, 1984.

\bibitem{Roth86}
A.~E. Roth.
\newblock On the allocation of residents to rural hospitals: a general property
  of two-sided matching markets.
\newblock {\em Econometrica}, 54(2):425--427, 1986.

\bibitem{schrijver2000combinatorial}
A.~Schrijver.
\newblock A combinatorial algorithm minimizing submodular functions in strongly
  polynomial time.
\newblock {\em Journal of Combinatorial Theory, Series B}, 80(2):346--355,
  2000.

\end{thebibliography}
\end{document}